\documentclass[conference]{IEEEtran}
\usepackage{amsthm}
\usepackage{url}
\usepackage[cp1255]{inputenc}
\usepackage{amsfonts}
\usepackage{amssymb}
\usepackage{amsmath}
\usepackage{verbatim}
\usepackage{graphics}
\usepackage{graphicx}
\usepackage{bbm}
\usepackage{cancel}
\usepackage{fancyhdr}
\usepackage{cite}

\let\ORIGRightarrow=\Rightarrow

\let\Rightarrow=\ORIGRightarrow

\usepackage[usenames]{color}

\usepackage{algorithmic}
\usepackage{algorithm}

\newcommand{\reals}{\mathbb{R}}

\newcommand{\beq}{\begin{equation}}
\newcommand{\eeq}{\end{equation}}

\newcommand{\beqa}{\begin{eqnarray}}
\newcommand{\eeqa}{\end{eqnarray}}

\newcommand{\norm}[1]{\left\|#1\right\|}
\newcommand{\set}[1]{\left\{#1\right\}}

\newcommand{\paren}[1]{\left(#1\right)}

\newcommand{\ul}{\underline}

\usepackage[font=bf]{subfig}
\usepackage{mathdots}
\usepackage[colorinlistoftodos]{todonotes}
\usepackage{marginnote}
\usepackage{booktabs}
\usepackage{multirow}

\DeclareGraphicsExtensions{.eps}
\graphicspath{{../figures/}}

\newtheorem{lemm}{Lemma}
\newtheorem{exmpl}{Example}

\begin{document}

\title{On Universal Sensor Registration
}

\author{\IEEEauthorblockN{Daniel Sigalov}
\IEEEauthorblockA{Rafael -- Advanced Defense Systems\\Israel\\
Email: danielsi@rafael.co.il}
\and
\IEEEauthorblockN{Aharon Gal}
\IEEEauthorblockA{Rafael -- Advanced Defense Systems\\Israel\\
Email: aharong@rafael.co.il}
\and
\IEEEauthorblockN{Boaz Vigdor}
\IEEEauthorblockA{Rafael -- Advanced Defense Systems\\Israel\\
Email: boazv@rafael.co.il}
}

\markboth{Submitted to IEEE Trans. on Signal Processing}%
{Shell \MakeLowercase{\textit{et al.}}: Bare Demo of IEEEtran.cls for Journals}

\maketitle

\begin{abstract}
We present a simple approach for sensor registration in target tracking applications. The proposed method uses targets of opportunity and, without making assumptions on their dynamical models, allows simultaneous calibration of multiple three- and two-dimensional sensors. Whereas for two-sensor scenarios only relative registration is possible, in practical cases with three or more sensors unambiguous absolute calibration may be achieved. The derived algorithms are straightforward to implement and do not require tuning of parameters. The performance of the algorithms is tested in a numerical study.
\end{abstract}


\section{Introduction}\label{section:intro}
Sensor calibration is a task of major importance in sensor fusion applications such as multi-sensor target tracking. The goal of the procedure is estimating and compensating for systematic, sensor related errors. These include additive or multiplicative biases in range, azimuth or elevation measurements of a radar, positional misplacement errors, and angular misalignment errors. Ignoring such errors, or failing to accurately compensate for their presence might lead to significant performance degradation of the target tracking system. For example, the data association module might fail correlating measurements from different sensors representing the same target. Consequently, a ``ghost'' track will be initiated. Alternatively, successful measurement association might come at the expense of the tracking filter performance since the model of the latter typically does not take into account systematic errors, which will, consequently, be interpreted as large innovations.

Various aspects of the sensor calibration problem have drawn much focus of the information fusion community.
In~\cite{helmick1993removal}, assuming small misalignment errors, a slave sensor was calibrated relatively to the master sensor using an extended Kalman filter. A similar, Kalman-based approach was taken in~\cite{nabaa1999solution} where a set of sensor was calibrated relatively to a master sensor. The authors also mentioned an unsuccessful attempt to perform an absolute calibration. An independent, yet related line of research is summarized in~\cite{fortunati2011least,fortunati2012least}, where a least-squares (LS) approach was taken for relative and absolute sensor registration. Under standard assumptions of small registration errors, that allow linearization of the nonlinear equation tying biases and target dynamics, the authors reported non-efficiency of the proposed algorithms. Sensitivity of the estimators in the above contributions was addressed in~\cite{tian2014robust}.
An approach of bias calibration using partial Kalman filter data was recently proposed in~\cite{taghavi2016practical}. A comprehensive literature survey summarizing most of the Kalman filter-based methods may be found in~\cite{topland2016joint}.

A different approach to data registration, somewhat overlooked by the classical information fusion community, may be found in a recent series of papers~\cite{chaudhury2015global,khoo2016non,sanyal2017registration}. Partially motivated by computer vision applications, the authors use semidefinite programming framework to achieve data/sensor registration of 3D sensors.

It is evident from the above literature survey, that the sensor registration problem has been addressed only partially in the past. Specifically, absolute registration of two-dimensional sensors in three-dimensional space still remains an interesting and open problem. In this paper we address the problem of calibrating angular misalignment errors of a group of two- or threes-dimensional sensors, 
and propose an approach that allows absolute (as opposed to relative) calibration of three or more sensors in all practical scenarios. 

The contribution of the present work is threefold. First, we derive simple calibration algorithms that do not require tuning of paraments and make no assumptions on the dynamical models of the targets under consideration. The second contribution is the absolute calibration capability of three-dimensional sensors that was addressed only partially in the information fusion community past. The final and major contribution is the absolute calibration capability of two-dimensional sensors in 3D space which, to the best of our knowledge, has not been addressed elsewhere before.

The remainder of the paper is organized as follows. In Section~\ref{section:problem} we formally state the problem. In Section~\ref{section:algorithm} we gradually derive the algorithms starting with the simplest case of relative calibration of two 3D sensors and concluding with absolute calibration of an arbitrary number of 2D sensors. In Section~\ref{section:numerical} we present a comprehensive numerical study for assessing the performance of the derived algorithm. Concluding remarks are made in Section~\ref{section:concluding}.


\section{Problem Formulation}
\label{section:problem}
Let $\mathcal{O}$ be an arbitrarily chosen origin of a north-east-down (NED) cartesian coordinate system.
We consider $S$ sensors at precisely known locations \begin{align}\label{Eq:problem:biased:sensor_positions}
\ul{\ell}_s\triangleq(x_{0,s},y_{0,s},z_{0,s})^T,\,s=1,\ldots,S. \end{align}
The sensors may be either three- or two-dimensional. In the former case, each sensor measurement carries the range, the azimuth and the elevation of the target, all of which are computed with respect to the local NED coordinated system $\mathcal{O}_s$ centered at the sensor location. For a two-dimensional sensor, only azimuth and elevation measurements are made.

Each sensor is characterized by a sensor-specific angular misalignment bias captured by a deterministic (but unknown) rotation matrix $\tilde{A}_s$. Mathematically, this may be formulated as follows. Assuming that the true position of a target, relative to $\mathcal{O}$, is $\ul{p}_0\triangleq(x_0,y_0,z_0)^T$, the biased position of this target in the sensor's local NED coordinate system having origin at $\mathcal{O}_s$ is given by
\begin{align}\label{Eq:problem:biased:position}
\ul{p}_s&\triangleq(x_s,y_s,z_s)= \tilde{A}_s(\ul{p}_0-\ul{\ell}_s).
\end{align}
The considered bias represents unknown calibration errors in the angular sensor alignment.
As mentioned, in this work, we do not address other types of biases such as misplacement errors, and additive/multiplicative sensor biases.

Consequently, for a three-dimensional sensor, the measurement of the target is
\begin{align}\label{Eq:problem:sensor:measurement}
\ul{m}_s\triangleq
\begin{pmatrix}{\rm rng}_s\\{\rm az}_s\\{\rm el}_s
\end{pmatrix}
=\begin{pmatrix}
\sqrt{ x_s^2+ y_s^2+ z_s^2}+n_{r,s}\\
\arctan{( y_s/ x_s)+n_{a,s}}\\
\arctan{( z_s/\sqrt{ x_s^2+ y_s^2})+n_{e,s}}
\end{pmatrix},
\end{align}
where $n_{r,s}$, $n_{a,s}$, and $n_{e,s}$ are measurement noises in, respectively, range, azimuth and elevation. These are taken to be independent, Gaussian, with zero mean and known variances $\sigma_r^2$, $\sigma_a^2$, and $\sigma_e^2$. For two-dimensional sensors the above remains valid except that the range measurement is absent in~\eqref{Eq:problem:sensor:measurement}.

Consider next a set of $n$ distinct target positions
\begin{align}\label{Eq:problem:set:positions:global}
\set{\ul{p}_0^i,\,i=1,\ldots,n}.
\end{align}
These may refer to different appearances of the same target or represent different objects either at the same or at different times. In any event, we do not assume temporal dependence between the elements of the set. In particular, no underlying dynamical model is presumed to be valid. Suppose that the above targets are observed by $S$ sensors. The above set of positions translates, for sensor $s=1,\ldots,S$, into the following set of vectors in $\reals^3$ in the local coordinate system $\mathcal{O}_s$
\begin{align}\label{Eq:problem:set:positions:local}
\set{\ul{p}_{s}^i,\,i=1,\ldots,n},\,s=1,\ldots,S.
\end{align}
Consequently, for each sensor, a set of $n$ measurements is generated
\begin{align}\label{Eq:problem:set:measurements}
\set{\ul{m}_{s}^i,\,i=1,\ldots,n},\,s=1,\ldots,S.
\end{align}
Here $\ul{m}_{s}^i$ represents a measurement of a physical object $\ul{p}_0^i$ generated by sensor $s$. Recall that $\ul{p}_0^i\in\reals^3$ in the global coordinate system with origin $\mathcal{O}$, $\ul{p}_s^i\in\reals^3$ in the local coordinate system with origin $\mathcal{O}_s$, and $\ul{m}_{s}^i$ is a noisy observation of the local version $\ul{p}_s^i$ carrying azimuth, elevation, and, potentially, range data. We further assume that data association and time synchronisation have been performed externally. In other words, $\ul{m}_{1}^i,\ul{m}_{2}^i,\ldots,\ul{m}_{S}^i$ represent the same target instance at the same time. These are not very restrictive assumptions since bias calibration is typically performed as a preliminary procedure using well separated targets of opportunity. In addition, time synchronisation may be achieved by, e.g., time-interpolation.

The goal of this paper may now be stated informally as follows. Given $S$ sets of measurements defined in~\eqref{Eq:problem:set:measurements}, we aim at estimating the individual rotation matrices $\tilde{A}_s,\,s=1,\ldots,S$ defined in Eq.~\ref{Eq:problem:biased:position}. We note that, a-priori, estimability, existence or uniqueness of the individual rotation matrices are not obvious. Before discussing the specific conditions, we proceed with formally formulating the objective. Observing~\eqref{Eq:problem:biased:position} one readily obtains the following expression for $\ul{p}_0^i$ that holds for all $s=1,\ldots,S$ and for all $i=1,\ldots,n$
\begin{align}\label{Eq:problem:biased:position:inverse}
\ul{p}_0^i&=A_s\ul{p}_s^i+\ul{\ell}_s,
\end{align}
where we introduced the notation $A_s\triangleq\tilde{A}_s^T$.
We thus consider the following minimization problem in order to obtain $A_s,\,s=1,\ldots,S$.
\begin{align}\label{Eq:problem:minimization:objective}
&\min_{A_1,\ldots,A_S}\sum_{s,t}\sum_{i=1}^n\norm{A_{s}\ul{p}_{s}^i+\ul{\ell}_{s}-\paren{A_{t}\ul{p}_{t}^i+\ul{\ell}_{t}}}^2\\\label{Eq:problem:minimization:constraint}
&{\rm s.t.}\,A_{1},\ldots,A_{S}\,\,{\rm rotation\, matrices},
\end{align}
where $s,t\in\set{1,\ldots,S}$, $s\neq t$, and $\norm{\ul{x}}$ denotes the $L^2$-norm of $\ul{x}$. We emphasize that $\ul{\ell}_{s}$ and $\ul{\ell}_{t}$ are the known locations of the sensors $s$ and $t$, respectively, ${\ul{p}_{s}^i}$ and ${\ul{p}_{t}^i}$ are vectors computed from the raw sensor measurements. The only unknowns in the above set of optimization problems are the sensor-specific rotation matrices. Note that~\eqref{Eq:problem:minimization:objective} is not a standard least-squares (LS) problem since the optimization domain is constrained to be the set of rotation matrices (orthogonal matrices with determinant $1$).

\section{Algorithm Derivation}
\label{section:algorithm}
In this section we derive and discuss the algorithm for the estimation of the sensor-specific rotation matrices as defined in~\eqref{Eq:problem:minimization:objective} and~\eqref{Eq:problem:minimization:constraint}. We begin the discussion with the easiest case of two 3D sensors and a single unknown matrix and proceed gradually to the most complex case of an arbitrary number of 2D sensors. In all algorithms in the sequel all the rotation matrices are initialized as identity matrices.

\subsection{Two 3D Sensors, Relative Calibration}
Let $S=2$ and assume that both sensors deliver three-dimensional measurements comprising range, azimuth, and elevation data as defined in~\eqref{Eq:problem:sensor:measurement}. Moreover, assume that there are no angular misalignment errors in the second sensor, namely, $A_2=I_{3\times3}$, where $I_{3\times3}$ is the $3\times3$ identity matrix. This setting may also be interpreted as angularly calibrating one sensor relatively to another (possibly also biased) sensor. In this case, the considered problem~\eqref{Eq:problem:minimization:objective}-\eqref{Eq:problem:minimization:constraint} reduces to
\begin{align}\label{Eq:algorithm:minimization:objective:2sensors:3d}
&\min_{A_{1}}\sum_{i=1}^n\norm{A_{1}\ul{p}_{1}^i+\ul{\ell}_{1}-\paren{\ul{p}_{2}^i+\ul{\ell}_{2}}}^2\\\label{Eq:algorithm:minimization:constraint:2sensors:3d}
&{\rm s.t.}\,A_{1}\,{\rm rotation\, matrix}.
\end{align}
This is a standard Wahba's problem~\cite{wahba1965least} which may be solved optimally by a variety of algorithms~\cite{markley1988attitude,arun1987least}. In the sequel, we refer to a routine solving the following formulation of the Wahba's problem
\begin{align}\label{Eq:algorithm:wahba:objective}
&\min_{A}\sum_{i=1}^n\norm{A\ul{x}^i-\ul{y}^i}^2\\\label{Eq:algorithm:wahba:constraint}
&{\rm s.t.}\,A\,{\rm rotation\, matrix},
\end{align}
as $\verb!Wahba!(\{\ul{x}^i\},\{\ul{y}^i\})$. Several possible implementations of the algorithm may be found in~\cite{markley1988attitude}.
The procedure for the relative calibration of angular misalignment errors in the case of two 3D, noiseless sensors is summarized in Alg.~\ref{Alg1:2rdr:3d:relative}.
\renewcommand{\algorithmicrequire}{\textbf{Input:}}
\renewcommand{\algorithmicensure}{\textbf{Output:}}
\begin{algorithm}[tbh]\caption{Relative Calibration of $2$ 3D Sensors}
\begin{algorithmic}[1]
\REQUIRE{$\set{\ul{m}_{s}^i,\,i=1,\ldots,n},\,s=1,2$.}%
\STATE{Compute $\{\ul{p}_{s}^i,\,i=1,\ldots,n\},\,s=1,2$.}\label{alg1:step:1}
\STATE{Compute $\{\ul{\tilde{p}}_{2}^i\triangleq\ul{p}_{2}^i+\ul{\ell}_{2}-\ul{\ell}_{1},\,i=1,\ldots,n\}$.}\label{alg1:step:2} %
\STATE{Compute $A_1=\verb!Wahba!(\{\ul{{p}_{1}}^i\},\{\ul{\tilde{p}_{2}}^i\})$.}\label{alg1:step:3}
\ENSURE{$A_1$.}%
\end{algorithmic}
\label{Alg1:2rdr:3d:relative}
\end{algorithm}
$ $\\
Step $1$ of Alg.~\ref{Alg1:2rdr:3d:relative} is a simple transformation of a measurement comprising range, azimuth, and elevation data into cartesian coordinate system centered at the sensor. Step $2$ is a translation of the vectors defined by the $\mathcal{O}_1$ coordinate system into the one defined by $\mathcal{O}_2$. In the absence of noise, two pairs of measurements suffice to find the required rotation. However, since measurement noise is inevitable, as defined in~\eqref{Eq:problem:sensor:measurement}, we need to modify the above algorithm by computing Step $1$ using noisy data, meaning the resulting sets $\{\ul{p}_{s}^i,\,i=1,\ldots,n\},\,s=1,2$ are noisy approximations of the actual cartesian target positions. Although, in principle, two pairs of measurements will still result in a valid rotation matrix, a larger number of such pairs will be required in order to improve the noise-robustness of the estimate.

\subsection{Two Heterogeneous Sensors, Relative Calibration}
\label{section:TwoHetrogeneous}
We proceed with the case of $S=2$ sensors where one of the sensors delivers, as before, three-dimensional measurements, and the other  generates two-dimensional data comprising (noisy) azimuth and elevation.
In this setting one cannot apply directly the procedure described in Alg.~\ref{Alg1:2rdr:3d:relative} since the two-dimensional sensor's measurement does not have a unique three-dimensional representation of the target position. To comply with the notation of the previous case we assume that sensor $s=2$ is the bias-free, three-dimensional sensor, while sensor $s=1$ generates two-dimensional data and has an unknown angular misalignment bias captured by the rotation matrix $A_1$.

Although a two-angle measurement does not correspond to a cartesian target position, it is equivalent to a direction vector that can be used to find the desired rotation matrix by modifying the objective~\eqref{Eq:algorithm:minimization:objective:2sensors:3d} as follows

\begin{align}\label{Eq:algorithm:minimization:objective:2sensors:23d}
&\min_{A_{1}}\sum_{i=1}^n\norm{A_{1}\ul{q}_{1}^i-\ul{q}_{2}^i}^2\\\label{Eq:algorithm:minimization:constraint:2sensors:23d}
&{\rm s.t.}\,A_{1}\,{\rm rotation\, matrix},
\end{align}
where $\ul{q}_{1}^i$ is the (noisy) direction vector (with norm $1$) that may be computed directly from the 2D sensor data,
\begin{align}
\label{Eq:algorithm:minimization:objective:2sensors:23d:q2}
\ul{q}_{2}^i&\triangleq\frac{\ul{\tilde{p}}_{2}^i}{\|\ul{\tilde{p}}_{2}^i\|},
\end{align}
and $\ul{\tilde{p}}_{2}^i=\ul{p}_{2}^i+\ul{\ell}_{2}-\ul{\ell}_{1}$.

In other words, instead of finding the optimal rotation between two sets of measurements, in the present case, we find the optimal rotation between the direction vectors resulting from the two-dimensional sensor data and from the translated versions of the three-dimensional sensor data. The complete routine is summarized in Alg.~\ref{Alg2:2rdr:23d:relative}.
\renewcommand{\algorithmicrequire}{\textbf{Input:}}
\renewcommand{\algorithmicensure}{\textbf{Output:}}
\begin{algorithm}[tbh]\caption{Relative Calibration of $2$ Heterogeneous Sensors}
\begin{algorithmic}[1]
\REQUIRE{$\set{\ul{m}_{s}^i,\,i=1,\ldots,n},\,s=1,2$.}%
\STATE{Compute $\{\ul{q}_{1}^i,\,i=1,\ldots,n\}$.}\label{alg2:step:1}
\STATE{Compute $\{\ul{p}_{2}^i,\,i=1,\ldots,n\}$.}\label{alg2:step:2}
\STATE{Compute $\{\ul{\tilde{p}}_{2}^i\triangleq\ul{p}_{2}^i+\ul{\ell}_{2}-\ul{\ell}_{1},\,i=1,\ldots,n\}$.}\label{alg2:step:3} %
\STATE{Compute $\{\ul{q}_{2}^i\triangleq\ul{\tilde{p}}_{2}^i/\|\ul{\tilde{p}}_{2}^i\|,\,i=1,\ldots,n\}$.}\label{alg2:step:4}
\STATE{Compute $A_1=\verb!Wahba!(\{\ul{{q}}_{1}^i\},\{\ul{{q}}_{2}^i\})$.}\label{alg2:step:5}
\ENSURE{$A_1$.}%
\end{algorithmic}
\label{Alg2:2rdr:23d:relative}
\end{algorithm}
$ $\\

\subsection{Two 3D Sensors, Absolute Calibration}
We are now ready to present the idea for the joint absolute calibration of an arbitrary number of three-dimensional sensors. We accomplish this task in two steps. In the present subsection, we consider only two such sensors and present the routine, based on the solution of the Wahba's problem, that generates valid rotation matrices for each of the sensors. We generalize the idea to an arbitrary number of sensors in the following subsection.

Consider the optimization problem~\eqref{Eq:problem:minimization:objective}-\eqref{Eq:problem:minimization:constraint} for $S=2$:
\begin{align}\label{Eq:algorithm:minimization:objective:2sensors}
&\min_{A_{1},A_{2}}\sum_{i=1}^n\norm{A_{1}\ul{p}_{1}^i+\ul{\ell}_{1}-\paren{A_{2}\ul{p}_{2}^i+\ul{\ell}_{2}}}^2\\\label{Eq:algorithm:minimization:constraint:2sensors}
&{\rm s.t.}\,A_{1},A_{2}\,{\rm rotation\, matrices},
\end{align}

While solving optimally, for both $A_1$ and $A_2$, may seem like a non-trivial task, recall that performing the optimization for either $A_1$ or $A_2$ separately is easy and the optimal routine is given in Alg.~\ref{Alg1:2rdr:3d:relative}. We thus consider the following approach motivated by the alternating least-squares (ALS) method~\cite{jain2013low,hastie2015matrix}. First, both $A_{1}$ and $A_{2}$ are initialized as identity matrices. Then, iteratively, two optimization steps are performed -- the Wahba's problem is solved with respect to $A_1$ while holding $A_2$ fixed at its latest value and, consequently, the Wahba's problem is solved with respect to $A_2$ while holding $A_1$ fixed at its latest value. The complete routine is summarized in Alg.~\ref{Alg3:2rdr:3d:absolute}.

\begin{algorithm}[tbh]\caption{Absolute Calibration of $2$ 3D Sensors}
\begin{algorithmic}[1]
\REQUIRE{$\set{\ul{m}_{s}^i,\,i=1,\ldots,n},\,s=1,2$.}
\STATE{Compute $\{\ul{p}_{s}^i,\,i=1,\ldots,n\},\,s=1,2$.}\label{alg3:step:1}
\STATE{Initialize ${A}_{1}^0=I_{3\times3}$, ${A}_{2}^0 =I_{3\times3}$, $j=1$.}\label{alg3:step:2}
\REPEAT\label{alg3:step:3}
\STATE{Compute $\{\ul{\tilde{p}}_{2}^i
\triangleq{A}_{2}^{j-1}\ul{p}_{2}^i+\ul{\ell}_{2}-\ul{\ell}_{1},\,i=1,\ldots,n\}$.}\label{alg3:step:4} %
\STATE{Compute $A_1^j=\verb!Wahba!(\{\ul{{p}}_{1}^i\},\{\ul{\tilde{p}}_{2}^i\})$.}\label{alg3:step:5}
\STATE{Compute $\{\ul{\tilde{p}}_{1}^i
\triangleq{A}_{1}^{j}\ul{p}_{1}^i+\ul{\ell}_{1}-\ul{\ell}_{2},\,i=1,\ldots,n\}$.}\label{alg3:step:6} %
\STATE{Compute $A_2^j=\verb!Wahba!(\{\ul{{p}}_{2}^i\},\{\ul{\tilde{p}}_{1}^i\})$.}\label{alg3:step:7}
\STATE{$j=j+1$}\label{alg3:step:8}
\UNTIL{Stopping criteria are met.} \label{alg3:step:9}
\ENSURE{$A_1, A_2$.}%
\end{algorithmic}
\label{Alg3:2rdr:3d:absolute}
\end{algorithm}
$ $\\

Note that Step~\ref{alg3:step:4} (and~\ref{alg3:step:6}) only rotate (using the currently known rotation matrix) and translate (using the known sensor locations) one set of vectors to the coordinate system of the second set as a preliminary step to the actual rotation matrix computation in Step~\ref{alg3:step:5} (and~\ref{alg3:step:7}).

The convergence of the algorithm is addressed in Lemma~\ref{lemma:convergence:2rdr}.
\begin{lemm}\label{lemma:convergence:2rdr}
Algorithm~\ref{Alg3:2rdr:3d:absolute} converges.
\end{lemm}
\begin{proof}Since $\verb!Wahba!(\{\ul{x}^i\},\{\ul{y}^i\})$ minimizes the cost $\sum_{i=1}^n\norm{A\ul{x}^i-\ul{y}^i}^2$, it is easy to see that the sequence of values of the objective is non-increasing and bounded from below and, therefore, converges.
\end{proof}
Unfortunately, the lemma only guaranties convergence of the cost to a local optimum. It is possible that the corresponding values of the rotation matrices will alternate between several values all resulting in the same cost of~\eqref{Eq:algorithm:minimization:objective:2sensors}. To gain better understanding of the reasoning behind Alg.~\ref{Alg3:2rdr:3d:absolute} 
we consider the following synthetic example.

\begin{exmpl}\label{exmpl:2rdr:planar} Consider the coplanar setup presented in Fig.~\ref{fig:algorithm:two3d:radars:plane:setup}.
\begin{figure}[h!t!]\centering
{\includegraphics[width=0.45\textwidth]{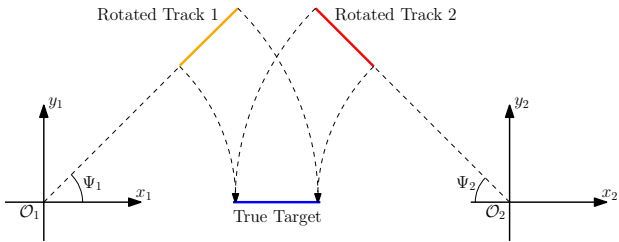}\label{fig:algorithm:two3d}}\\
\caption{Scenario geometry of Example~\ref{exmpl:2rdr:planar}}
\label{fig:algorithm:two3d:radars:plane:setup}
\end{figure}
\begin{figure*}[h!t!]\centering
{\includegraphics[width=0.24\textwidth, trim={0.2cm 0cm 0cm 0.2cm}, clip]{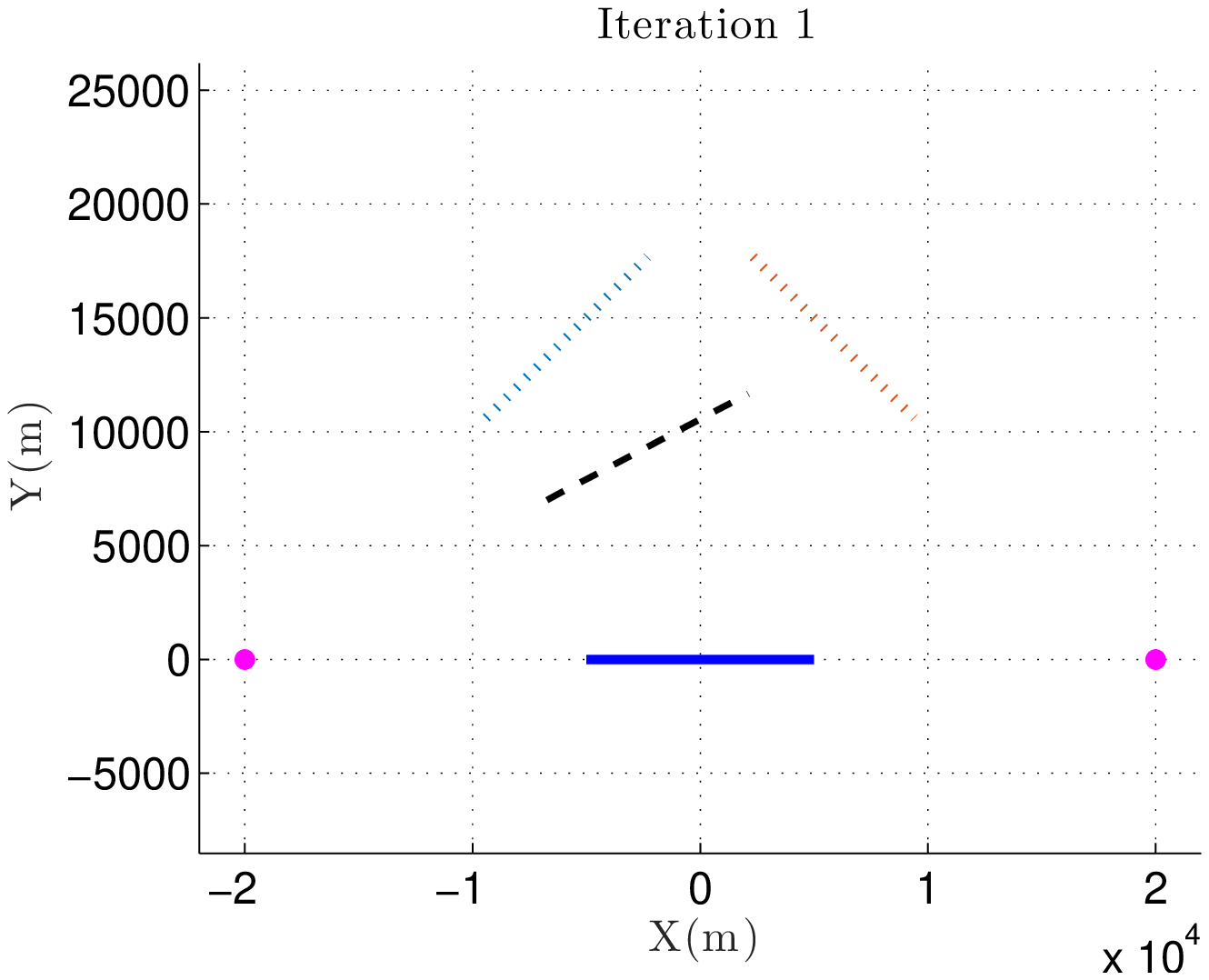} \label{fig:algorithm:two3d:iter1}}
{\includegraphics[width=0.24\textwidth, trim={0.2cm 0cm 0cm 0.2cm}, clip]{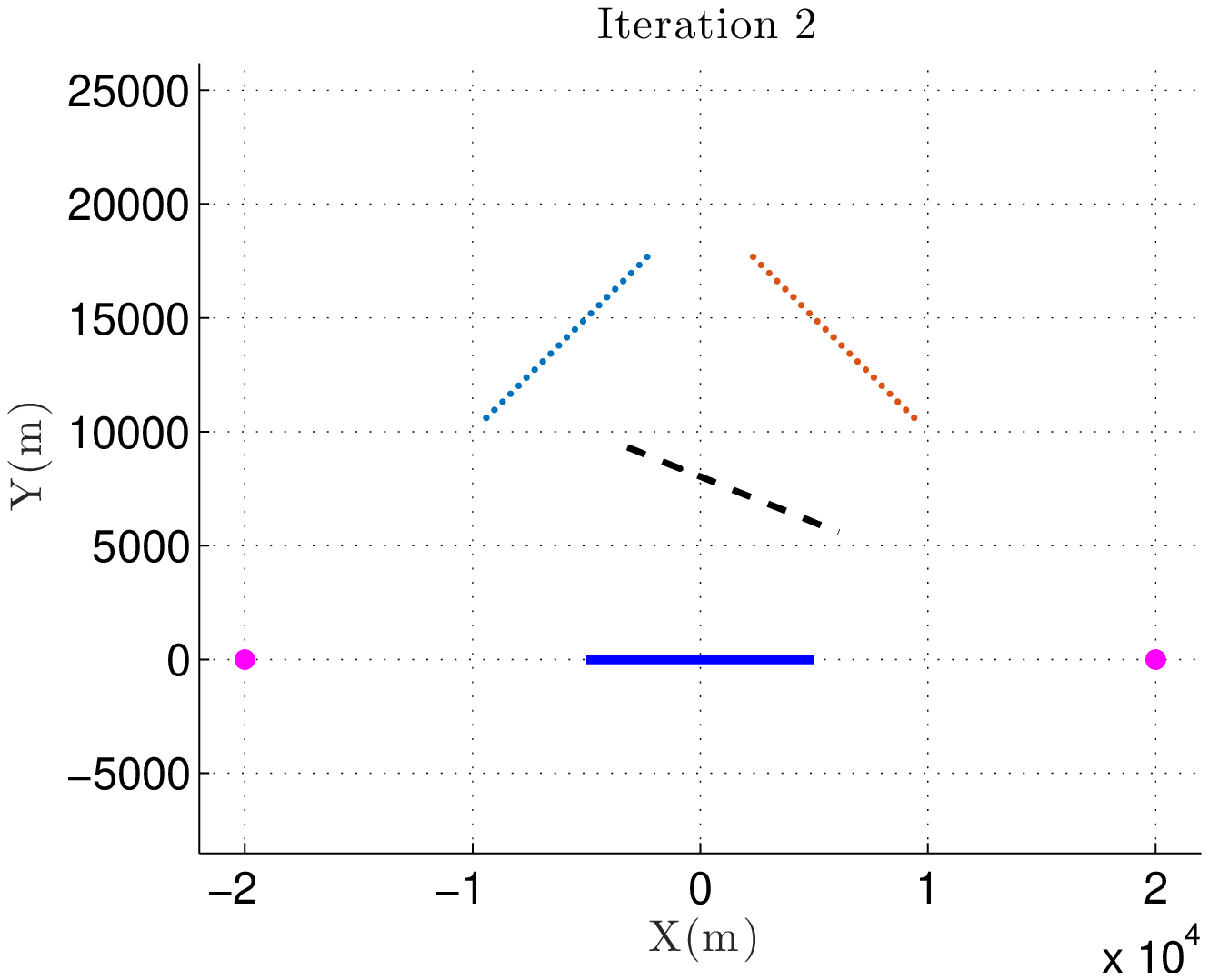} \label{fig:algorithm:two3d:iter2}}
{\includegraphics[width=0.24\textwidth, trim={0.2cm 0cm 0cm 0.2cm}, clip]{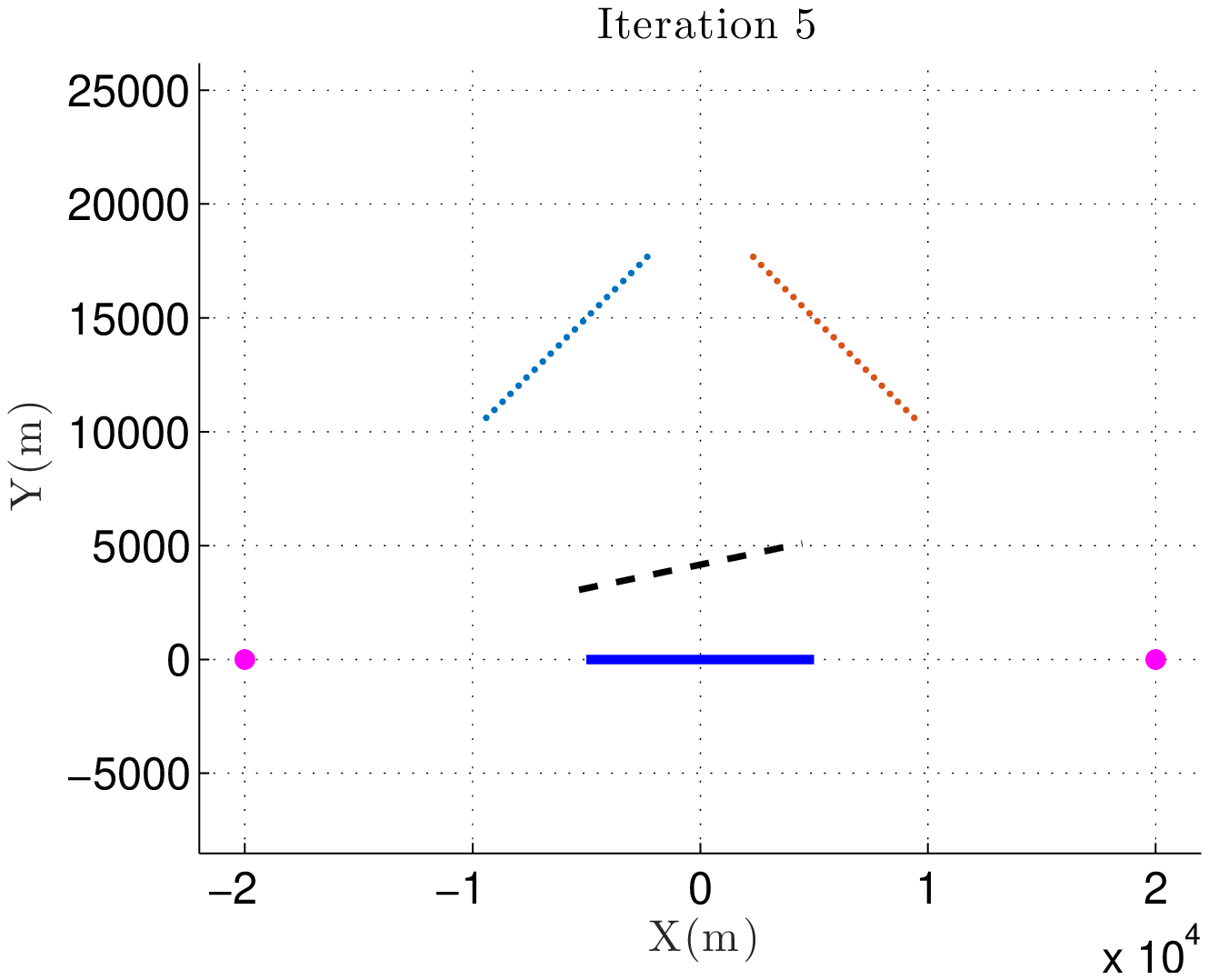} \label{fig:algorithm:two3d:iter5}}
{\includegraphics[width=0.24\textwidth, trim={0.2cm 0cm 0cm 0.2cm}, clip]{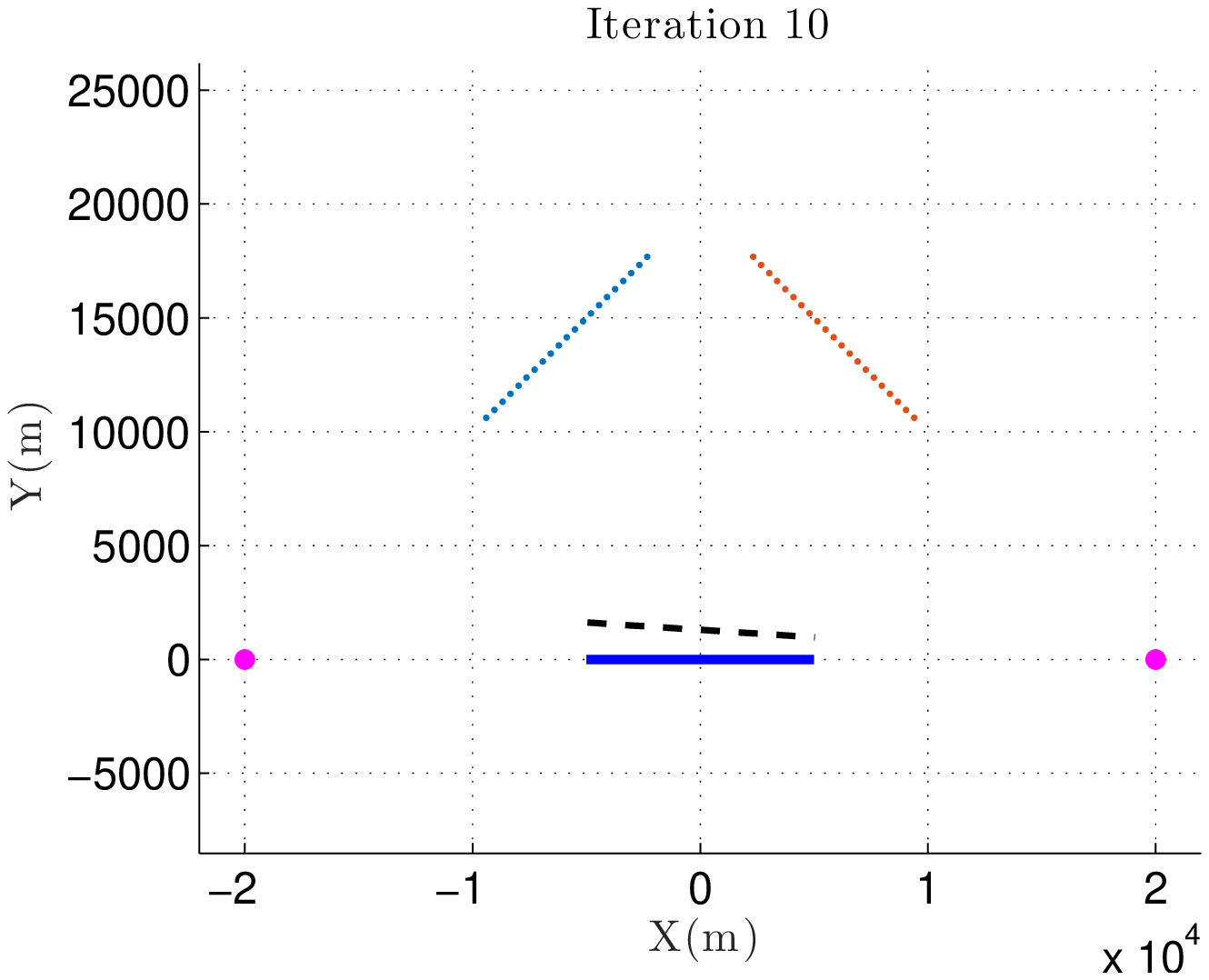} \label{fig:algorithm:two3d:iter10}}
\caption{Selected iterations of Alg.~\ref{Alg3:2rdr:3d:absolute} for the setup of Example~\ref{exmpl:2rdr:planar}. Magenta circles represent the sensors. Solid lines represent the true target and dotted lines are the rotated tracks. Dashed lines represent the reconstruction of the original trajectory.
The rotated tracks converge to the true target thus revealing the originally unknown rotation angles.}
\label{fig:algorithm:two3d:radars:plane:evolution}
\end{figure*}
The axes $x_1-y_1$ of the coordinate system $\mathcal{O}_1$ as well as $x_2-y_2$ of $\mathcal{O}_2$ define a horizontal plane which also contains the true target presented in blue. We assume that the only angular misalignment of the two sensors is about $z_1$ and $z_2$ -- the two vertical axes perpendicular to, respectively, $x_1-y_1$ and  $x_2-y_2$. The corresponding misalignment angles are $\Psi_1$ and $\Psi_2$. Consequently, the true target is represented by the sensors 1 and 2 as, respectively, ``Rotated Track 1'' and ``Rotated Track 2''. Although the two rotation angles are unknown, it is intuitively clear that a minimum of~\eqref{Eq:algorithm:minimization:objective:2sensors} will be attained if ``Rotated Track 1'' is rotated clockwise by $\Psi_1$ about $z_1$ and ``Rotated Track 2'' is rotated counterclockwise by $\Psi_2$ about $z_2$ until both coincide with the blue line representing the true target, thus revealing the originally unknown rotation angles. Note, that the considered setup is indifferent to rotations about the $x_1$ (or $x_2$) axes meaning that biases in this direction cannot be estimated. The evolution of Alg.~\ref{Alg3:2rdr:3d:absolute} is presented in Fig.~\ref{fig:algorithm:two3d:radars:plane:evolution}.
\end{exmpl}

\subsection{$S\geq3$ 3D Sensors, Absolute Calibration}
\label{algorithm:General3Dcase}
We are now at the position to address the case of a general number of three-dimensional sensors. The idea is essentially the same as in the case of $2$ sensors with the exception that, at each iteration, a sequence of relative optimizations is performed such that all $S$ rotation matrices undergo (relative) optimization. Thus, the alternating least-squares approach may now be though of as ``multi-element alternating (constrained) least squares''. Specifically, at each iteration of the proposed algorithm we perform the Wahba's optimization between every two pairs of matrices. The specific order of the sensors, to be followed at each iteration, is user-dependent and its impact on the performance is beyond the scope of this paper. In the numerical study section we used the following order of the $S$ sensors: $1\to2, 2\to1, 1\to3, 3\to1, \ldots, 1\to S,S\to1,2\to3, 3\to2, \ldots 2\to S,S\to2, \ldots, S-1\to S, S\to S-1$, where $s_i\to s_j$ refers to performing relative alignment of sensor $s_i$ to sensor $s_j$. The complete routine is summarized in Alg.~\ref{Alg4:Nrdr:3d:absolute}.


\begin{algorithm}[tbh]\caption{Absolute Calibration of $S$ 3D Sensors}
\begin{algorithmic}[1]
\REQUIRE{$\set{\ul{m}_{s}^i,\,i=1,\ldots,n},\,s=1,2,\ldots,S$.}
\STATE{Compute $\{\ul{p}_{s}^i,\,i=1,\ldots,n\},\,s=1,2,\ldots,S$.}\label{alg4:step1}
\STATE{Initialize ${A}_{s}^0=I_{3\times3},\,s=1,2,\ldots,S$, $j=1$.}\label{alg4:step2}
\REPEAT\label{alg4:step:3}
\FOR{$s=2$ \TO $S$} \label{alg4:step:4}
\FOR{$t=1$ \TO $s-1$} \label{alg4:step:5}
\STATE{Compute $\{\ul{\tilde{p}}_{s}^i
\triangleq{A}_{s}^{j-1}\ul{p}_{s}^i+\ul{\ell}_{s}-\ul{\ell}_{t},\,i=1,\ldots,n\}$.}\label{alg4:step:6} %
\STATE{Compute $A_{t}^j=\verb!Wahba!(\{\ul{{p}}_{t}^i\},\{\ul{\tilde{p}}_{s}^i\})$.}\label{alg4:step:7}
\STATE{Compute $\{\ul{\tilde{p}}_{t}^i
\triangleq{A}_{t}^{j}\ul{p}^i_{t}+\ul{\ell}_{t}-\ul{\ell}_{s},\,i=1,\ldots,n\}$}\label{alg4:step:8} %
\STATE{Compute $A_s^j=\verb!Wahba!(\{\ul{{p}}_{s}^i\},\{\ul{\tilde{p}}_{t}^i\})$.}\label{alg4:step:9}
\ENDFOR
\ENDFOR
\STATE{$j=j+1$}\label{alg4:step:10}
\UNTIL{Stopping criteria are met.}\label{alg4:step:11}
\ENSURE{$A_1,\ldots, A_S$.}%
\end{algorithmic}
\label{Alg4:Nrdr:3d:absolute}
\end{algorithm}



\subsection{Two 2D Sensors, Absolute Calibration}
\label{section:algorithm:Two2Dcase}
In the present subsection we proceed with the case of $S=2$ two-dimensional sensors that deliver (noisy) azimuth and elevation data.
Since two-dimensional measurements do not have a unique three-dimensional representation of the target position, one cannot apply directly the procedure described in Alg.~\ref{Alg3:2rdr:3d:absolute}. However, two or more such sensors, placed at different locations, allow us to estimate the three-dimensional target position by means of triangulation. This is the idea behind the algorithms proposed in the sequel.

Consider again the optimization problem~\eqref{Eq:algorithm:minimization:objective:2sensors} as well as the constraint~\eqref{Eq:algorithm:minimization:constraint:2sensors}. Recall that $\ul{p}_1^i$ and $\ul{p}_2^i$ are vectors in $\reals^3$ representing noisy target positions and generated from the three-dimensional raw sensor measurements. In the present setting we do not have direct access to these vectors and, thus, the optimization problem formulation has to be modified. One possible modification is as follows:
\begin{align}\label{Eq:algorithm:minimization:objective:2sensors:triang}
&\min_{A_{1},A_{2}}\sum_{i=1}^n\norm{A_{1}\ul{\hat{p}}_{1}^i+\ul{\ell}_{1}-\paren{A_{2}\ul{\hat{p}}_{2}^i+\ul{\ell}_{2}}}^2\\\label{Eq:algorithm:minimization:constraint:2sensors:triang}
&{\rm s.t.}\,A_{1},A_{2}\,{\rm rotation\, matrices}\\\label{Eq:algorithm:minimization:triangulation:2sensors:triang}
&(\ul{\hat{p}}_{1}^i,\, \ul{\hat{p}}_{2}^i)=g(m_1^i,m_2^i).
\end{align}
Here, $g:\reals^2\times\reals^2\to\reals^3\times\reals^3$ is a function that, given two pairs of angle measurements, $m_1^i,m_2^i$, returns two cartesian vectors $\ul{\hat{p}}_{1}^i,\,\ul{\hat{p}}_{2}^i$ that represent target positions. The latter are computed using the raw azimuth and elevation data together with range datum obtained from a triangulation of the angle measurements. Note that, up to the missing range datum in the present case and a slight abuse of notation, $m_s^i$ is defined similarly to Eq.~\ref{Eq:problem:sensor:measurement}. The computation of $(\ul{\hat{p}}_{1},\, \ul{\hat{p}}_{2})=g(m_1,m_2)$ is given in the Appendix.

Inspecting the formulation~\eqref{Eq:algorithm:minimization:objective:2sensors:triang}-\eqref{Eq:algorithm:minimization:triangulation:2sensors:triang}
we make the following observation. Given the rotation matrices, computing the cartesian target positions is straightforward. On the other hand, given $(\ul{\hat{p}}_{1}^i,\, \ul{\hat{p}}_{2}^i)$, the problem reduces to finding the required rotation matrices which we have already solved in the previous section. We thus adopt an iterative approach, in which the steps of computing the cartesian target positions and estimating the rotation matrices are performed alternatingly. The resulting solution is summarized in Alg.~\ref{Alg6:2rdr:2d:absolute}.

\begin{algorithm}[tbh]\caption{Absolute Calibration of $2$ 2D Sensors}
\begin{algorithmic}[1]
\REQUIRE{$\set{\ul{m}_{s}^i,\,i=1,\ldots,n},\,s=1,2$.}
\STATE{Initialize ${A}_{1}^0=I_{3\times3}$, ${A}_{2}^0 =I_{3\times3}$, $j=1$.}\label{alg6:step:1}
\REPEAT\label{alg6:step:2}
\STATE{Compute $\{(\ul{\hat{p}}_{1}^i,\ul{\hat{p}}_{2}^i)=g(\ul{m}_{1}^i,\ul{m}_{2}^i),i=1,\ldots,n\}$.}\label{alg6:step:3}
\STATE{Compute $\{\ul{\tilde{p}}_{2}^i
\triangleq{A}_{2}^{j-1}\ul{p}_{2}^i+\ul{\ell}_{2}-\ul{\ell}_{1},\,i=1,\ldots,n\}$.}\label{alg6:step:4} %
\STATE{Compute $A_1^j=\verb!Wahba!(\{\ul{\hat{p}}_{1}^i\},\{\ul{\tilde{p}}_{2}^i\})$.}\label{alg6:step:5}
\STATE{Compute $\{\ul{\tilde{p}}_{1}^i
\triangleq{A}_{1}^{j}\ul{p}_{1}^i+\ul{\ell}_{1}-\ul{\ell}_{2},\,i=1,\ldots,n\}$.}\label{alg6:step:6} %
\STATE{Compute $A_2^j=\verb!Wahba!(\{\ul{\hat{p}}_{2}^i\},\{\ul{\tilde{p}}_{1}^i\})$.}\label{alg6:step:7}
\STATE{Compute $\{\ul{m}_{s}^i,\,i=1,\ldots,n\},\,s=1,2$.}\label{alg6:step:8}
\STATE{$j=j+1$}\label{alg6:step:9}
\UNTIL{Stopping criteria are met.} \label{alg6:step:10}
\ENSURE{$A_1, A_2$.}%
\end{algorithmic}
\label{Alg6:2rdr:2d:absolute}
\end{algorithm}
Note that in order to properly account for the updated rotation matrices after Step~\ref{alg6:step:7}, we use them to recompute $m_1^i$ and $m_2^i$.

\subsection{General 2D Case}
\label{section:algorithm:General2Dcase}
Our final step is the absolute registration of an $S\geq3$ 2D sensors. The algorithm is a straightforward integration of the ideas behind Algs.~\ref{Alg4:Nrdr:3d:absolute} and~\ref{Alg6:2rdr:2d:absolute}. Specifically, at the beginning of each iteration we first compute the estimates of the cartesian target positions based on the currently estimated values of the rotation matrices of each sensor. In the next step, the ``multi-element alternating constrained least-squares'' problem introduced in Subsection~\ref{algorithm:General3Dcase} is performed. The output of this step are the refined estimates of the rotation matrices. We note in passing that the order of sensors, to be followed in the sequence of relative calibrations, remains, as before, user-dependent. We consider the same order of sensors as in the 3D case: $1\to2, 2\to1, 1\to3, 3\to1, \ldots, 1\to S,S\to1,2\to3, 3\to2, \ldots 2\to S,S\to2, \ldots, S-1\to S, S\to S-1$, where $s_i\to s_j$ refers to performing relative alignment of sensor $s_i$ to sensor $s_j$. The complete routine is summarized in Alg.~\ref{Alg7:Nrdr:2d:absolute}.


\begin{algorithm}[tbh]\caption{Absolute Calibration of $S$ 2D Sensors}
\begin{algorithmic}[1]
\REQUIRE{$\set{\ul{m}_{s}^i,\,i=1,\ldots,n},\,s=1,2,\ldots,S$.}
\STATE{Initialize ${A}_{s}^0=I_{3\times3},\,s=1,2,\ldots,S$, $j=1$.}\label{alg7:step1}
\REPEAT\label{alg7:step:2}
\STATE{Compute
$\{(\ul{\hat{p}}_{1}^i,\ldots,\ul{\hat{p}}_{S}^i)=g(\ul{m}_{1}^i,\ldots,\ul{m}_{S}^i),\,i=1,\ldots,n\}$.}\label{alg7:step:3}
\FOR{$s=2$ \TO $S$} \label{alg7:step:6}
\FOR{$t=1$ \TO $s-1$} \label{alg7:step:7}
\STATE{Compute $\{\ul{\tilde{p}}_{s}^i
\triangleq{A}_{s}^{j-1}\ul{p}_{s}^i+\ul{\ell}_{s}-\ul{\ell}_{t},\,i=1,\ldots,n\}$.}\label{alg7:step:8} %
\STATE{Compute $A_{t}^j=\verb!Wahba!(\{\ul{\hat{p}}_{t}^i\},\{\ul{\tilde{p}}_{s}^i\})$.}\label{alg7:step:9}
\STATE{Compute $\{\ul{\tilde{p}}_{t}^i
\triangleq{A}_{t}^{j}\ul{p}^i_{t}+\ul{\ell}_{t}-\ul{\ell}_{s},\,i=1,\ldots,n\}$}\label{alg7:step:10} %
\STATE{Compute $A_s^j=\verb!Wahba!(\{\ul{\hat{p}}_{s}^i\},\{\ul{\tilde{p}}_{t}^i\})$.}\label{alg7:step:11}
\ENDFOR
\ENDFOR
\STATE{Compute $\{\ul{m}_{s}^i,\,i=1,\ldots,n\},\,s=1,2,\ldots,S$.}\label{alg7:step:12}
\STATE{$j=j+1$}\label{alg7:step:13}
\UNTIL{Stopping criteria are met.}\label{alg7:step:14}
\ENSURE{$A_1,\ldots, A_S$.}%
\end{algorithmic}
\label{Alg7:Nrdr:2d:absolute}
\end{algorithm}
Similarly to the 3D case, $g:\reals^2\times\cdots\times\reals^2\to\reals^3\times\cdots\times\reals^3$ is a function that, given $S$ pairs of angle measurements, $m_1^i,\ldots,m_2^i$, returns $S$ cartesian vectors $\ul{\hat{p}}_{1}^i,\ldots,\ul{\hat{p}}_{2}^i$ that represent target positions.

\subsection{Discussion}
Algs.~\ref{Alg4:Nrdr:3d:absolute} and~\ref{Alg7:Nrdr:2d:absolute} are the main results of the present section. We note that the algorithms are straightforward to implement and require no tuning parameters. The only user-dependent decision is the order of sensors for the alternating relative calibration. In addition, it is possible to perform Step~\ref{alg7:step:3} of Alg.~\ref{Alg7:Nrdr:2d:absolute} not only at the beginning of each iteration, but also inside the inner loop. The effect of such modification is beyond the scope of this paper. 

\section{Numerical Study}
\label{section:numerical}
In this section we demonstrate the performance of the derived algorithms in a numerical simulation.
We consider a single target of opportunity observed by a variety of sensors each having an angular misalignment bias. We emphasize that no assumptions are made on the dynamical model of the target and only raw sensor measurements are used to estimate the rotation biases. In the sequel, $\Psi$, $\Theta$, and $\Phi$ are, respectively, the yaw, pitch, and roll rotation angles.

\subsection{Simulation Setup}
We consider a single target moving in a 3D space. The target maintains a constant speed of $10$ (m/s) in the vertical direction and $100$ (m/s) in the horizontal plane, where it performs legs of constant velocity movement in the direction of the $x$-axes interleaved with coordinated turn legs. The target is observed by a varying number of $3$ to $10$ two- or three-dimensional randomly placed sensors. The top-view of the trajectory and the sensor positions are shown in Fig.~\ref{fig:numerical:scenario}.
\begin{figure}[h!t!]\centering
{\includegraphics[width=0.44\textwidth, trim=0cm 0cm 0cm 0cm]{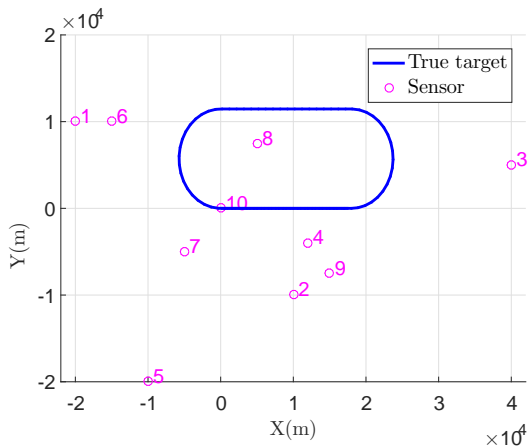}}\hspace{0.2cm}
\caption{Target trajectory and sensor locations.}
\label{fig:numerical:scenario}
\end{figure}

The sensors are assumed to be synchronous with a constant sampling rate of $0.1$ (Hz). Each sensor generates, every $10$ (secs), a noisy measurement comprising azimuth, elevation and possible range data. The duration of the scenario if $15$ minutes.

The measurement noises, defined in~\eqref{Eq:problem:sensor:measurement}, are taken to be zero-mean, independent Gaussian random variables with standard deviations $\sigma_r=10$ (m), $\sigma\triangleq\sigma_a=\sigma_e=3$ (mRad) unless stated otherwise. Sensitivity to these quantities as well as convergence properties are tested in the sequel.

\subsection{Single Run}
For a visual demonstration of the results we consider the above scenario with $3$ 3D sensors and exaggeratedly large misalignment errors. The true and estimated misalignment errors are summarized in Table~\ref{table:numerical:errors:single:run}. The corresponding biased and calibrated trajectories of each sensor, along with the true trajectory, are shown on the left- and right-hand side of Fig.~\ref{fig:numerical:bizarre}, respectively.
\begin{table}[h!t!]
\renewcommand{\arraystretch}{1.3}
\caption{The estimated bias angles of the single run example. The actually used values appear in the parentheses. 
}
\begin{center}
\begin{tabular}{l|c|c|c}
\toprule
& $\Psi$ (deg) & $\Theta$ (deg) & $\Phi$ (deg) \\
\hline
Sensor 1   & $10.003$ $(10)$ & $-9.990$ $(-10)$ & $9.969$ $(10)$  \\
\hline
Sensor 2 & $-10.016$ $(-10)$ & $9.978$ $(10)$ & $9.986$ $(10)$ \\
\hline
Sensor 3 & $9.969$ $(10)$ & $10.000$ $(10)$ & $-9.980$ $(-10)$ \\
\bottomrule
\end{tabular}
\label{table:numerical:errors:single:run}
\end{center}
\end{table}

\begin{figure*}[h!t!]\centering
{\includegraphics[width=0.43\textwidth, trim=0cm 0cm 0cm 0cm]{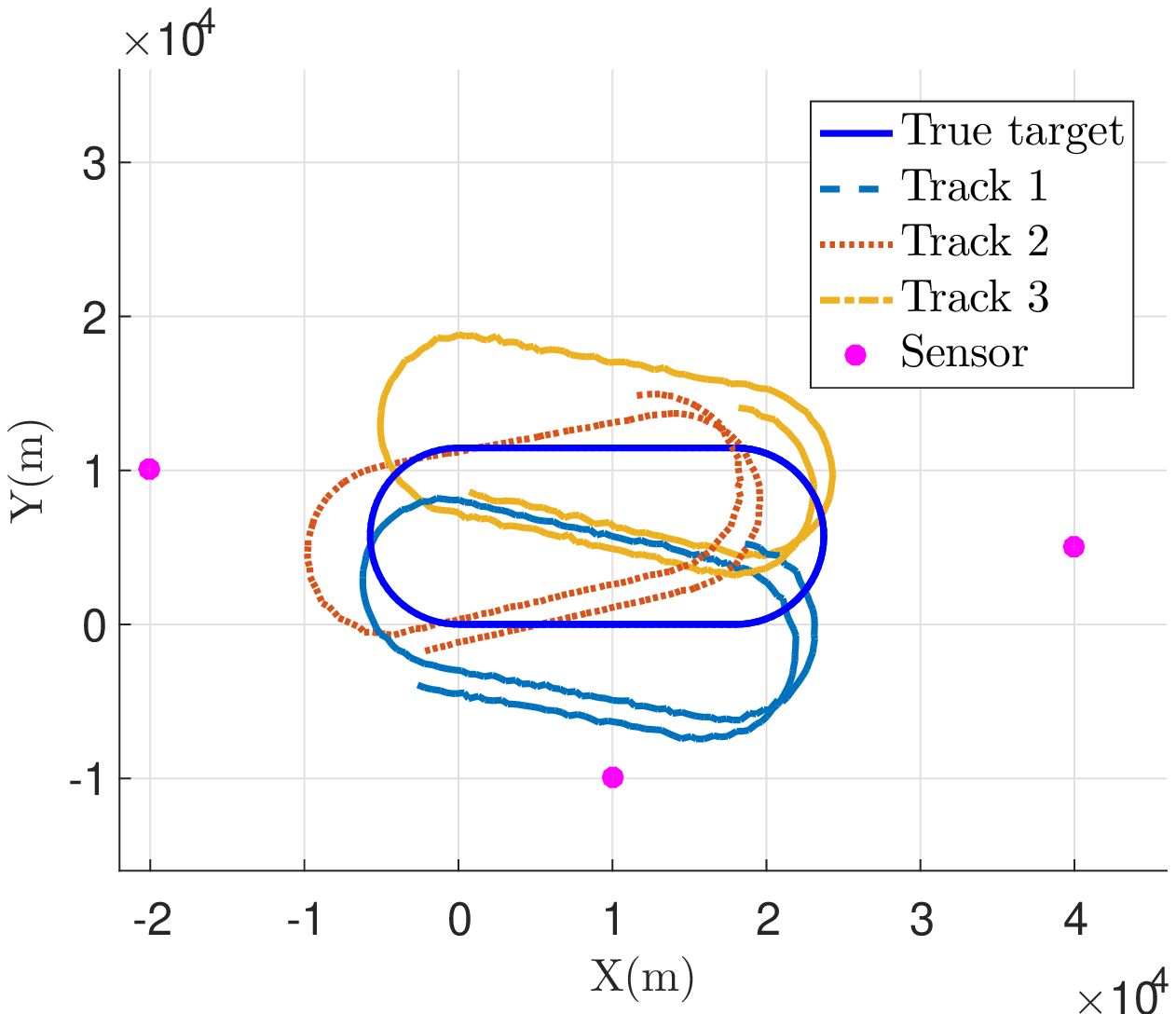}}\hspace{0.2cm}
{\includegraphics[width=0.43\textwidth, trim=0cm 0cm 0cm 0cm]{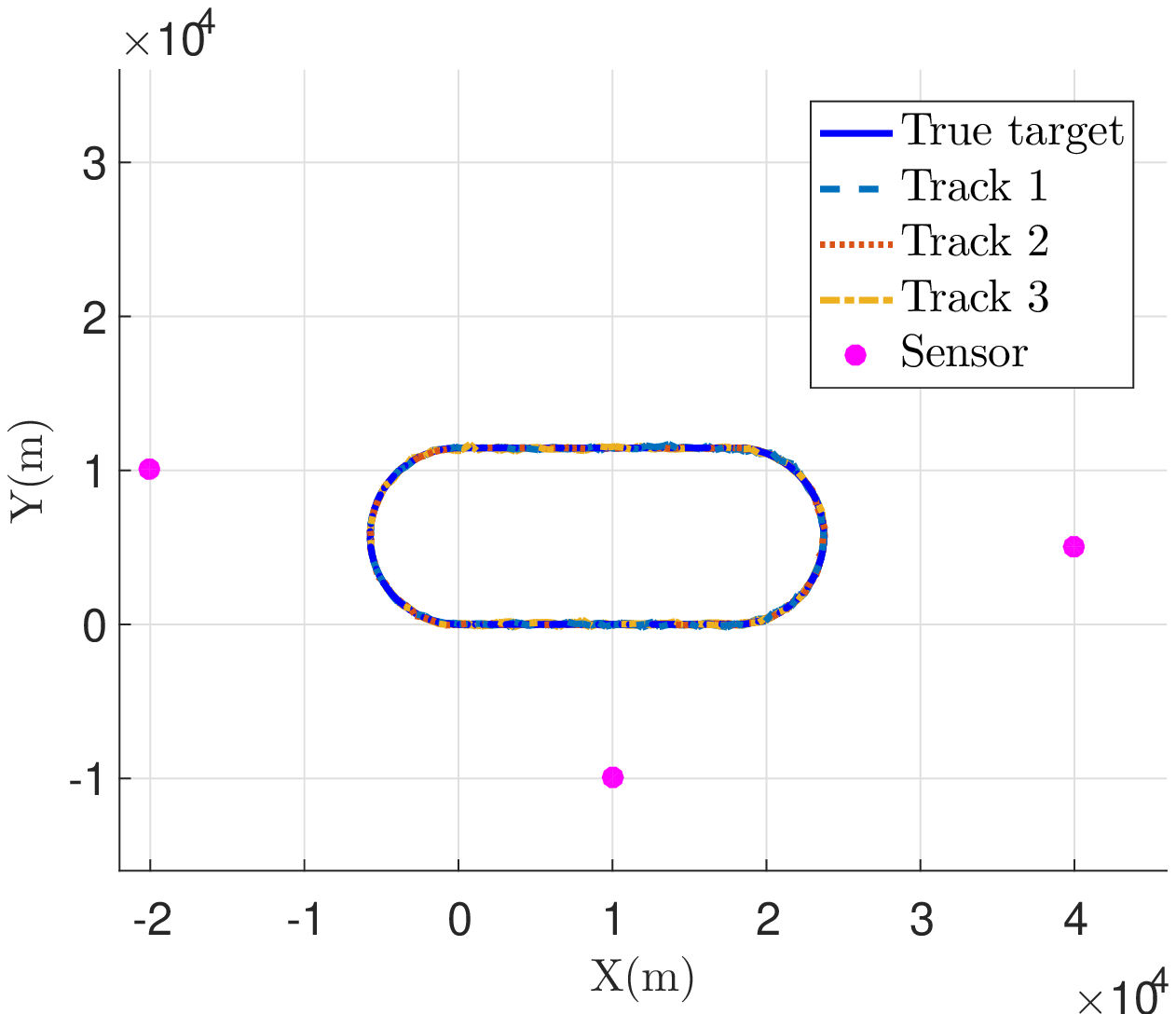}}
\caption{Sensor calibration example with $3$ 3D sensors. Trajectories before calibration (left). Trajectories after calibration (right). All tracks are rotated to the ground truth.}
\label{fig:numerical:bizarre}
\end{figure*}

\subsection{Performance Assessment}
We now present the results of a comprehensive numerical study for the performance assessment of the derived algorithms. We consider the same test trajectory and estimate the rotation matrices of the participating sensors whose number varies from $S=3$ to $S=8$. The bias of each sensor was taken to be between $1$ and $4$ degrees ($17$ to $70$ mRad) in azimuth, elevation and roll. The locations of the sensors are fixed as shown in Fig.~\ref{fig:numerical:scenario}. Each scenario was tested twice -- with two- and three-dimensional sensors. For each sensor configuration (type and number) we performed $50$ independent Monte Carlo (MC) realizations of the measurement noises. The resulting root mean-squared (RMS) estimation errors, averaged across the sensors are presented in Fig.~\ref{fig:numerical:2d3d:errors}.
\begin{figure}[h!]\centering
{\includegraphics[width=0.43\textwidth]{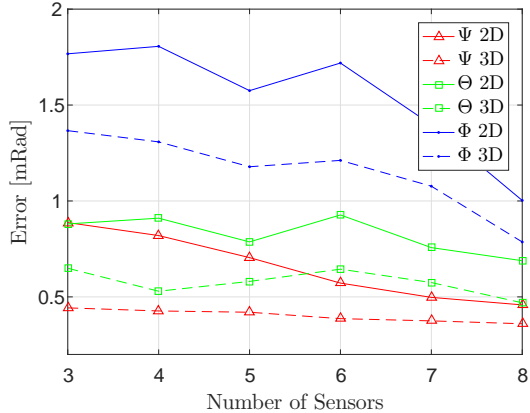}}\hspace{0.2cm}
\caption{RMS errors vs. the number of sensors. 2D and 3D scenarios represented by solid  and dashed lines, respectively. }
\label{fig:numerical:2d3d:errors}
\end{figure}
It is readily seen that both the number of sensors and the dimension of the measurements improve the estimation performance of the bias errors. In addition, note that, despite the $3$ (mRad) measurement noise, all estimation errors are bounded from above by $\approx2$ (mRad) and, in some cases, the errors drops to less than $0.5$ (mRad).
In other words, the original bias errors of up to $70$ (mRad) have been reduced by a factor of $\approx35$ or more.

\subsection{Convergence and Sensitivity Analysis}
We next test the convergence properties of the algorithms as well as their robustness to measurement noise variance and well as the number of samples. Recall that, according to Lemma~\ref{lemma:convergence:2rdr}, the basic algorithm converges for $S=2$ sensors. This desired property, however, is not obvious for a larger number of sensors.

We present the estimation errors as a function of the iteration number for $S=4$ sensors in Fig.~\ref{fig:numerical:convergence}.
The case of 3D sensors is shown on the left-hand side, and the case of 2D sensors is shown on the right-hand side. The bias values as well as measurement noise variance are the same as in the previous subsection.
\begin{figure*}[h!t!]\centering
{\includegraphics[width=0.42\textwidth, trim=0cm 0cm 0cm 0cm, clip]{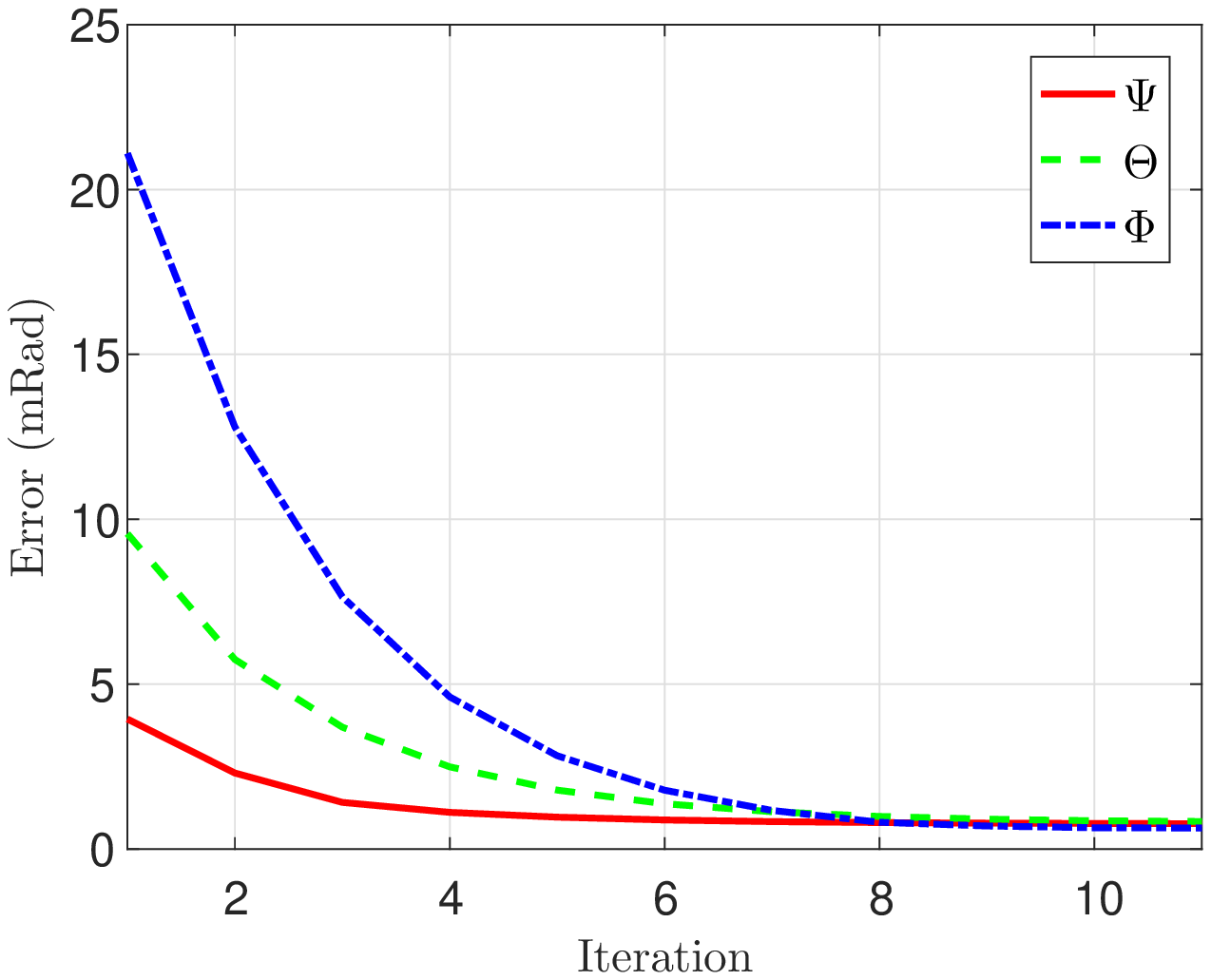}}\hspace{0.2cm}
{\includegraphics[width=0.42\textwidth, trim=0cm 0cm 0cm 0cm, clip]{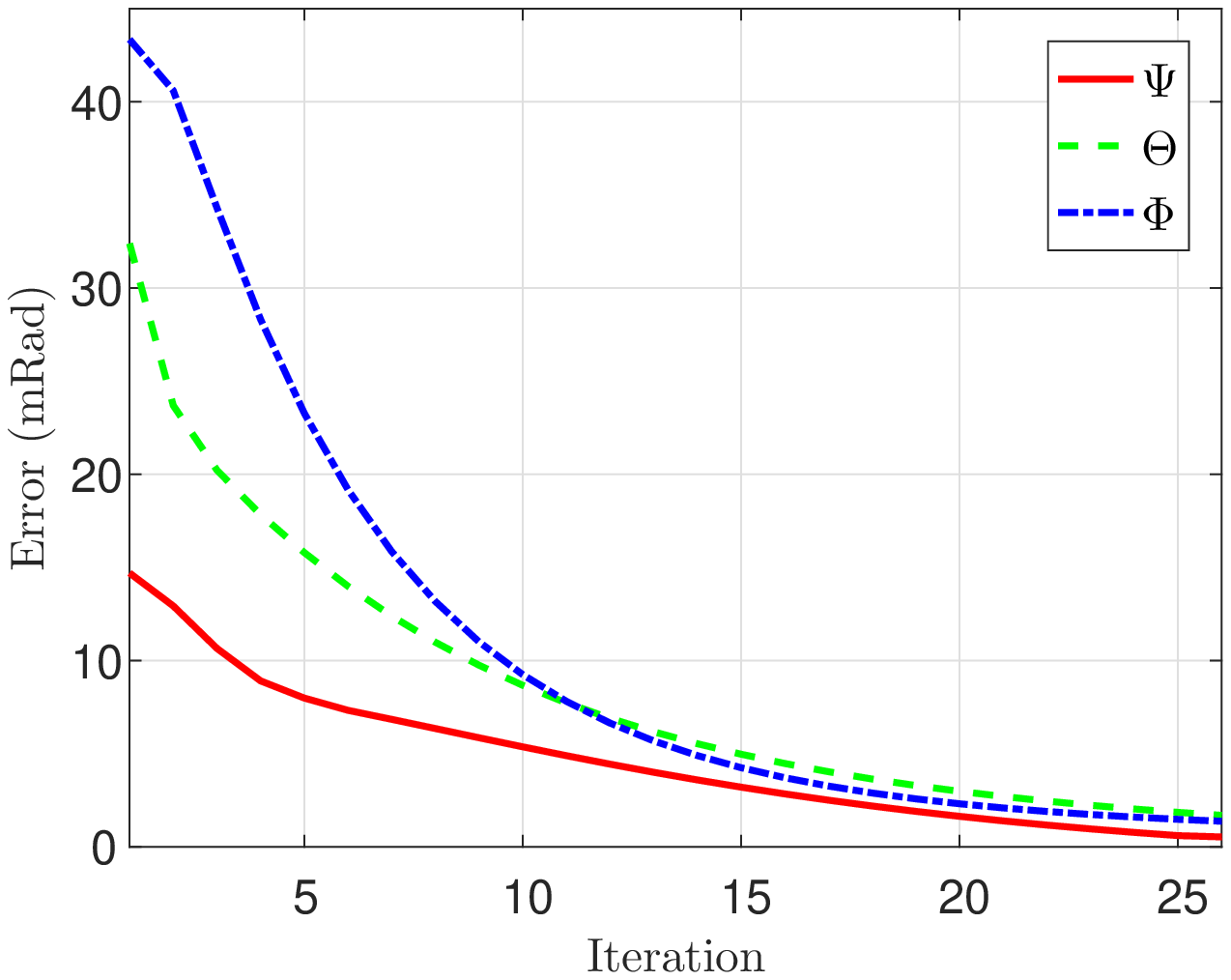}}
\caption{Error convergence vs. iteration number. 3D sensors (left). 2D sensors (right).}
\label{fig:numerical:convergence}
\end{figure*}
As expected, the convergence for the 3D case is significantly faster than that for the 2D case. While the estimates in the former case converge in less than 10 iterations, as much as 25 are needed in the latter case. This is easily explained by the introduction of the triangulation step to the algorithm which strongly depends on the uncompensated biases. On the other hand, bias estimation degrades when triangulation is inaccurate.

Our final study is the sensitivity analysis of the algorithm for 2D sensor calibration to the number of samples and measurement noise. On the left-hand side of Fig.~\ref{fig:numerical:sensitivity} we present the estimation errors for various numbers of measurements and $S=4$ sensors. For clarity of the presentation, the horizontal axes is taken be on a logarithmic scale. When changing the number of measurements, we keep, however, the target-sensor geometry nearly unaltered. That is, in order to reduce the number of measurements we simply down-sample the original trajectory. It is readily seen that with as few as $10$ measurements, the errors in the pitch and yaw angles drop to $\approx2$ (mRad) and the error in the roll angle reduce to $\approx4$ (mRad). The dependence of the estimation errors on the measurement noise variance is nearly linear and is presented on the right-hand side of Fig.~\ref{fig:numerical:sensitivity}.

\begin{figure*}[h!t!]\centering
{\includegraphics[width=0.42\textwidth, trim=0cm 0cm 0cm 0cm, clip]{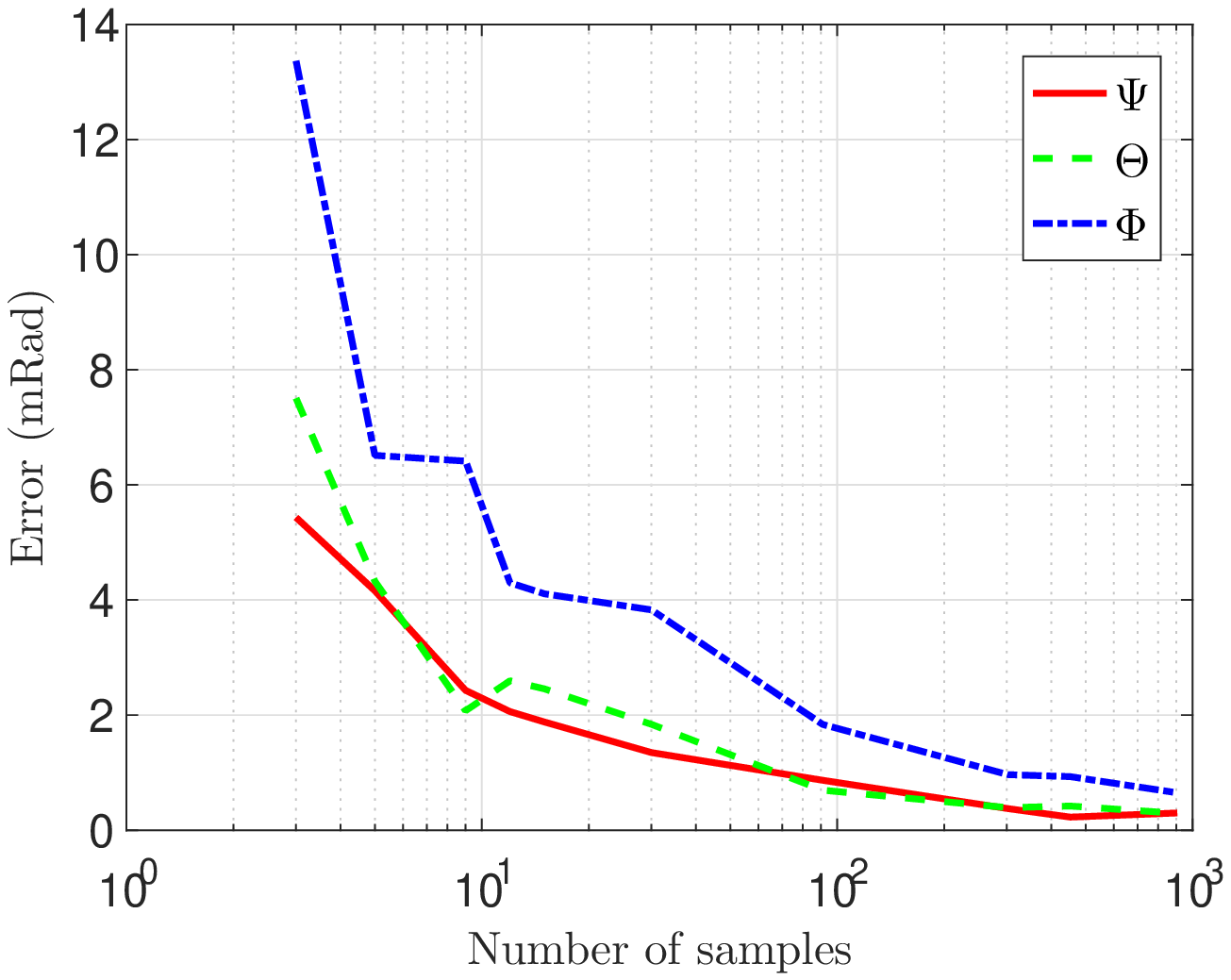}}\hspace{0.2cm}
{\includegraphics[width=0.42\textwidth, trim=0cm 0cm 0cm 0cm, clip]{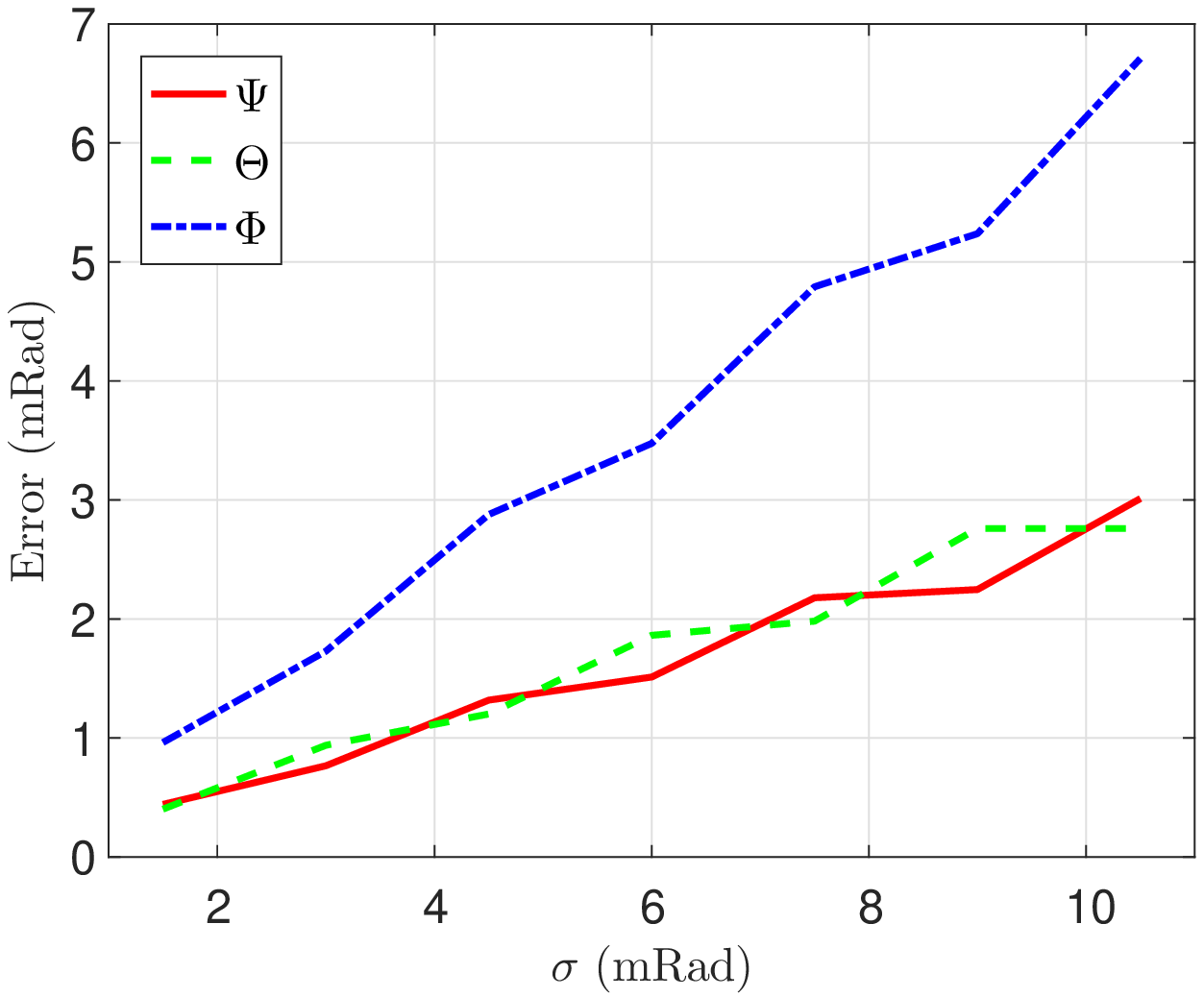}}
\caption{Sensitivity of the algorithm for 2D sensors. Errors vs. number of samples (left). Errors vs. noise std (right).}
\label{fig:numerical:sensitivity}
\end{figure*}

\section{Concluding Remarks}\label{section:concluding}
We presented a simple approach for sensor registration in target tracking applications. The proposed method uses targets of opportunity and, without making assumptions on their dynamical models, allows simultaneous calibration of multiple three- and two-dimensional sensors. The presented algorithms are simple to implement and free of tuning parameters. While it was shown that for two sensors, absolute calibration is possible only up to a degree of freedom, it follows from our numerical study that for $3$ or more sensors unambiguous absolute calibration is achievable.

Several directions may be explored to extend the results of this work. Scenarios with mixed 2D and 3D sensors may be simultaneously calibrated using the ideas presented in this paper. In addition, it is interesting to formulate analytical conditions under which absolute calibration is unambiguously achievable. It is conjectured here that as long as the sensors are not co-linear a unique solution exists. Finally, the performance of the algorithms depends on the sensor-target geometry. Investigation of favorable geometries as well as their characterization is of potential interest.


\section{Appendix}
\label{appendix}
\begin{itemize}
\item Compute $(x,y,z)$ -- the triangulation point of $m_s=({\rm az}_s,{\rm el}_s),\,s=1,2$ by 
solving the nonlinear LS problem
$
\arg\min_{(x,y,z)}\norm{\ul{f}(x,y,z)-({\rm az}_1\,{\rm el}_1\,{\rm az}_2\,{\rm el}_2)^T}^2,
$
where
\[
\ul{f}(x,y,z)=
\begin{pmatrix}
\arctan(\Delta y_1/\Delta x_1)\\
\arctan(\Delta z_1/\sqrt{\Delta x_1^2+\Delta y_1^2})\\
\arctan(\Delta y_2/\Delta x_2)\\
\arctan(\Delta z_2/\sqrt{\Delta x_2^2+\Delta y_2^2})
\end{pmatrix}
\]
and $\Delta x_s\triangleq x-x_{0,s}$, $\Delta y_s\triangleq y-y_{0,s}$, $\Delta z_s\triangleq z-z_{0,s}$.
\item Compute $r_s\triangleq\sqrt{\Delta x_2^2+\Delta y_2^2+\Delta z_2^2}$.
\item Compute $\ul{\hat{p}}_{s}$ from the polar representation $({\rm az}_s,{\rm el}_s,r_s)$.
\end{itemize}


\bibliographystyle{IEEEtran}
\bibliography{../IEEEabrv,../SigalovReferences}

%
%
%

\end{document}